\documentclass[aap,MSNbibl,dvips]{arximspdf}
\usepackage{mathbh}
\usepackage{graphicx}

\doi{10.1214/13-AAP991} 
\volume{25}
\issue{1}
\pubyear{2015}
\firstpage{150}
\lastpage{176}

\makeatletter

\newcommand{\essinf}{\mathop{\operatorname{ess\,inf}}}

\newcommand{\inter}{\mathop{\operatorname{int}}}

\newtheorem{theorem}{Theorem}[section]
\newtheorem{lemma}[theorem]{Lemma}
\newtheorem{corollary}[theorem]{Corollary}
\newtheorem{proposition}[theorem]{Proposition}
\newproclaim{definition}[theorem]{Definition}
\newproclaim{remark}[theorem]{Remark}
\newproclaim{Assumptions}[theorem]{Assumption}
\makeatother

\begin{document}
\begin{frontmatter}

\title{On an integral equation for the free-boundary of stochastic,
irreversible investment problems}
\runtitle{An integral equation for the free-boundary}

\begin{aug}
\author[A]{\fnms{Giorgio} \snm{Ferrari}\corref{}\ead[label=e1]{giorgio.ferrari@uni-bielefeld.de}\thanksref{T1}}
\runauthor{G. Ferrari}
\address[A]{Center for Mathematical Economics\\
Bielefeld University\\
Universit\"{a}tsstra{\ss}e 25\\
D-33615 Bielefeld\\
Germany\\
\printead{e1}}
\affiliation{Bielefeld University}
\end{aug}
\thankstext{T1}{Supported by the German Research Foundation (DFG) via
Grant Ri 1128-4-1,
\textsl{Singular Control Games: Strategic Issues in Real Options and
Dynamic Oligopoly under Knightian Uncertainty}.}

\received{\smonth{12} \syear{2012}}
\revised{\smonth{11} \syear{2013}}

%
\begin{abstract}
In this paper, we derive a new handy integral equation for the
free-boundary of infinite time horizon, continuous time, stochastic,
irreversible investment problems with uncertainty modeled as a
one-dimensional, regular diffusion $X$. The new integral equation
allows to explicitly find the free-boundary $b(\cdot)$ in some so far
unsolved cases, as when the operating profit function is not
multiplicatively separable and $X$ is a three-dimensional Bessel
process or a CEV process.
Our result follows from purely probabilistic arguments. Indeed, we
first show that $b(X(t))=l^{*}(t)$, with $l^{*}$ the unique optional
solution of a representation problem in the spirit of Bank--El Karoui
[\textit{Ann. Probab.} \textbf{32}
(2004) 1030--1067]; then, thanks to such an identification and the
fact that $l^{*}$ uniquely solves a backward stochastic equation, we
find the integral problem for the free-boundary.
\end{abstract}

%
\begin{keyword}[class=AMS]
\kwd[Primary ]{93E20}
\kwd{60G40}
\kwd[; secondary ]{91B70}
\kwd{60H25}
\end{keyword}

\begin{keyword}
\kwd{Integral equation}
\kwd{free-boundary}
\kwd{irreversible investment}
\kwd{singular stochastic control}
\kwd{optimal stopping}
\kwd{one-dimensional diffusion}
\kwd{Bank and El Karoui's representation theorem}
\kwd{base capacity}
\end{keyword}
\end{frontmatter}

\section{Introduction}
\label{introduction}

In this paper, we find a new integral equation for the free-boundary
$b(\cdot)$ arising in infinite time horizon, continuous time,
stochastic, irreversible investment problems of the form
%
\begin{equation}
\label{problemintro} \sup_{\nu}\mathbb{E} \biggl\{\int
_0^{\infty} e^{-r t} \pi
\bigl(X^{x}(t), y + \nu(t)\bigr) \,dt - \int_0^{\infty}
e^{- r t} \,d\nu(t) \biggr\},
\end{equation}
with $X^{x}$ regular, one-dimensional diffusion modeling market
uncertainty. The integral problem for $b(\cdot)$ is derived by means of
purely probabilistic arguments.
After having completely characterized the solution of the singular
control problem~(\ref{problemintro}) by some first-order conditions for
optimality and in terms of the \textsl{base capacity} process~$l^{*}$,
unique optional solution of a representation problem \`{a} la Bank--El
Karoui~\cite{BankElKaroui}, we show that $l^{*}(t)=b(X^{x}(t))$. Such
an identification, the strong Markov property and a beautiful result in
\cite{Borodin} on the joint law of a regular, one-dimensional diffusion
and its running supremum both stopped at an independent exponentially
distributed random time, lead to the integral equation for $b(\cdot)$
%
\begin{equation}
\label{integraleqintro} \psi_r(x)\int_{x}^{\overline{x}}
\biggl(\int_{\underline{x}}^z \pi _c
\bigl(y,b(z)\bigr) \psi_r(y) m(dy) \biggr) \frac{s(dz)}{\psi_r^2(z)} = 1.
\end{equation}
Here, $\pi_c(x,c)$ is the instantaneous marginal profit function,
$\underline{x}$ and $\overline{x}$ the endpoints of the domain of
$X^{x}$, $r$ the discount factor, $\mathcal{G}$ the infinitesimal
generator associated to $X^{x}$, $\psi_r(x)$ the increasing solution to
the ordinary differential equation $\mathcal{G}u = ru$ and $m(dx)$ and
$s(dx)$ the speed measure and the scale function measure of $X^{x}$,
respectively.
The rather simple structure of equation (\ref{integraleqintro}) allows
to explicitly find the free-boundary even in some nontrivial settings;
that is, for example, the case of $X^{x}$ given by a three-dimensional
Bessel process and Cobb--Douglas or CES (constant elasticity of
substitution) profits. Such a result appears here for the first time.

The connection between irreversible investment problems under
uncertainty, optimal stopping and free-boundary problems is well known
in the economic and mathematical literature (cf., e.g., the monography
by Dixit and Pyndick \cite{DixitPindyck}).
From the mathematical point of view, a problem of optimal irreversible
investment may be modeled as a ``monotone follower'' problem; that is, a
problem in which control strategies are nondecreasing stochastic
processes, not necessarily absolutely continuous with respect to the
Lebesgue measure as functions of the time. Work on ``monotone follower''
problems and their application to Economics started with the early
papers by Karatzas, Karatzas and Shreve, El Karoui and Karatzas (cf.
\cite{Karatzas81,KaratzasShreve84} and \cite
{KaratzasElKarouiSkorohod}), among others. These authors studied the
problem of optimally minimizing expected costs when the controlled
diffusion is a Brownian motion starting at $x \in\mathbb{R}$ tracked
by a nondecreasing process, that is, the monotone follower.
By relying on purely probabilistic arguments, they showed that one may
associate to such a singular stochastic control problem a suitable
optimal stopping problem whose value function $v$ is related to the
value function $V$ of the original control problem by $v=\frac
{\partial
}{\partial x}V$. Moreover, the optimal stopping time $\tau^{*}$ is such
that $\tau^{*}=\inf\{t \geq0\dvtx \nu^{*}(t) > 0\}$, with $\nu^{*}$ the
optimal singular control.
Later on, this kind of link has been established also for more
complicated dynamics of the controlled diffusion; that is the case, for
example, of a geometric Brownian motion \cite{KaratzasBaldursson}, or
of a quite general controlled It\^o diffusion (see \cite{Benth} and
\cite{BoetiusKohlmann}, among others).

Usually (see \cite{Chiarolla4,Chiarolla2,Kobila,AOksendal,Pham} and \cite{RiedelSu}, among others), the
optimal irreversible investment policy consists in waiting until the
shadow value of installed capital is below the marginal cost of
investment; on the other hand, the times at which the shadow value of
installed capital equals the marginal cost of investment are optimal
times to invest. It follows that from the mathematical point of view
one must find the region in which it is profitable to invest
immediately (the so-called ``action region'') and the region in which,
instead, it is optimal to wait (the so-called ``no-action region'' or
``continuation region''). The boundary between these two regions is the
free-boundary of the optimal stopping problem naturally associated to
the singular control one. The optimal investment is then the least
effort to keep the controlled process inside the closure of the
``continuation region;'' that is, in a diffusive setting, the local time
of the optimally controlled diffusion at the free-boundary.

In the last decade, many papers addressed singular stochastic control
problems by means of a first-order conditions approach (cf., e.g.,
\cite
{BankRiedel1,Bank,CF,CFR,RiedelSu} and
\cite{Steg}), not necessarily relying on any Markovian or diffusive
structure. The solution of the optimization problem is indeed related
to that of a representation problem for optional processes (cf. \cite
{BankElKaroui}): the optimal policy consists in keeping at time $t$ the
state variable always above the lower bound $l^{*}(t)$, unique optional
solution of a stochastic backward equation \`{a} la Bank--El Karoui
\cite
{BankElKaroui}.
Clearly, such a policy acts like the optimal control of singular
stochastic control problems as the original monotone follower problem
(see, e.g., \cite{Karatzas81} and \cite{KaratzasShreve84}) or, more
generally, irreversible investment problems (cf. \cite
{KaratzasBaldursson,Chiarolla2,Kobila} and \cite
{AOksendal}, among others). Therefore, in a diffusive setting, the
signal process $l^{*}$ and the free-boundary $b(\cdot)$ arising in
singular stochastic control problems must be linked.
In \cite{CF}, the authors studied a continuous time, singular
stochastic irreversible investment problem over a finite time horizon
and they showed that for a production capacity given by a controlled
geometric Brownian motion with deterministic, time-dependent
coefficients one has $l^{*}(t)=b(t)$.

In this paper, we aim to understand the meaning of the process $l^{*}$
for the whole class of infinite time horizon, irreversible investment
problems of type (\ref{problemintro}).
By means of a first-order conditions approach, we first find the
optimal investment policy in terms of the ``base capacity'' process
$l^{*}$ (cf. \cite{RiedelSu}, Definition $3.1$), unique optional
solution of a representation problem in the spirit of Bank--El Karoui
\cite{BankElKaroui}. That completely solves control problem (\ref
{problemintro}).
The policy to invest just enough to keep the production capacity above
$l^{*}(t)$ turns out to be the optimal investment strategy at time $t$.
The base capacity process defines therefore a desirable value of
capacity that the controller aims to maintain.
We show indeed that $l^{*}(t)=b(X^{x}(t))$, where~$b(\cdot)$ is the
free-boundary of the optimal stopping problem
%
\begin{equation}
\label{optimalstoppingintro} v(x,y)=\inf_{\tau\geq0}\mathbb{E} \biggl\{\int
_0^{\tau} e^{-rs}\pi _c
\bigl(X^{x}(s),y\bigr) \,ds + e^{-r \tau} \biggr\}
\end{equation}
associated to (\ref{problemintro}) (cf., e.g., \cite
{KaratzasBaldursson}, Lemma $2$).
Such an identification, together with the fact that $l^{*}$ uniquely
solves a backward stochastic equation [see (\ref{backward}) below],
yields a new integral equation for the free-boundary [cf. (\ref
{integraleqintro}) and also our Theorem~\ref{integraleqthm} below]. Our
equation does not rely on It\^o's formula and does not require any
smooth-fit property or a priori continuity of $b(\cdot)$ to be applied.
In this sense, it differs from that one could derive from the local
time--space calculus of Peskir for semimartingales on continuous
surfaces \cite{Peskir} (such approach has been used in the context of
stochastic, irreversible investment problems in \cite{Chiarolla2} and,
more recently, in \cite{DeAF2013} for a reversible, stochastic
investment problem).
Notice that for multiplicatively separable profit functions [i.e., $\pi
(x,c)=f(x)g(c)$, as in the Cobb--Douglas case] problem (\ref
{optimalstoppingintro}) may be easily reduced to the linearly
parameter-dependent optimal stopping problem $\sup_{\tau\geq
0}\mathbb
{E}\{e^{-r\tau}(u(X^{x}(\tau)) - k)\}$ completely solved in \cite
{BankBaumgarten} for a regular, one-dimensional diffusion $X$ (take
$u(x):=\mathbb{E}\{\int_0^{\infty}e^{-rs}f(X^{x}(s))\,ds\}$ and
$k:=1/g'(y)$ as a real parameter to obtain by the strong Markov
property $v(x,y)=g'(y)[u(x) -\sup_{\tau\geq0}\mathbb{E}\{e^{-r\tau
}(u(X^{x}(\tau)) - k)\}]$). In \cite{BankBaumgarten}, the free-boundary
in the $(x,k)$-plane is obtained in terms of the infimum of an
auxiliary function of one variable that can be determined from the
Laplace transforms of the level passage times of $X$. However, our
integral equation (\ref{integraleqintro}) is derived for very general
concave profit functions and can be analytically solved even in
nonseparable cases, as when the profit is of CES type (see Section~\ref{CESsubsection} below). This represents one of the main novelties of
this work.

The paper is organized as follows. Section~\ref{problem} introduces the
optimal control problem. In Section~\ref{optimalsolutionFB}, we find
the optimal investment strategy, we identify the link between the ``base
capacity'' process and the free-boundary and we derive the integral
equation for the latter one.
Finally, in Section~\ref{Examples}, we discuss some relevant examples,
as the case in which the economic shock $X^{x}$ is a geometric Brownian
motion, a three-dimensional Bessel process or a CEV process and the
profits are Cobb--Douglas or CES.

\section{The optimal investment problem}
\label{problem}

On a complete filtered probability space $(\Omega, \mathcal{F},
\mathbb
{P})$, with $\{\mathcal{F}_t, t \geq0\}$ the filtration generated by
an exogenous Brownian motion $\{W(t), t \geq0\}$ and augmented by
$\mathbb{P}$-null sets, consider the optimal irreversible investment
problem of a firm. The uncertain status of the economy is represented
by the one-dimensional, time-homogeneous diffusion $\{X^{x}(t), t \geq
0\}$ with state space $\mathcal{I} \subseteq\mathbb{R}$, satisfying
the stochastic differential equation (SDE)
%
\begin{equation}
\label{SDE} \cases{ %
dX^{x}(t) = \mu
\bigl(X^{x}(t)\bigr)\,dt + \sigma\bigl(X^{x}(t)\bigr)\,dW(t),
\vspace*{2pt}\cr
X^{x}(0)= x, }
\end{equation}
for some Borel functions $\mu\dvtx  \mathcal{I} \mapsto\mathbb{R}$ and
$\sigma\dvtx \mathcal{I} \mapsto(0,+\infty)$. We assume that $\mu$ and
$\sigma$ fulfill
%
\begin{equation}
\label{YamadaWatanabe} \cases{ %
\bigl|\mu(x)-\mu(y)\bigr| \leq
K|x-y|,
\vspace*{2pt}\cr
\bigl|\sigma(x)-\sigma(y)\bigr| \leq h\bigl(|x-y|\bigr),}
\end{equation}
for every $x,y \in\mathcal{I}$, and for some $K>0$ and $h\dvtx \mathbb{R}_+
\mapsto\mathbb{R}_+$ strictly increasing, such that $h(0) = 0$ and
%
\begin{equation}
\label{YamadaWatanabe2} \int_{(0,\varepsilon)}\frac{du}{h^2(u)} = \infty\qquad
\mbox{for every } \varepsilon>0.
\end{equation}
Hence, pathwise uniqueness holds for the SDE (\ref{SDE}) by the
Yamada--Watanabe theorem (cf. \cite{KaratzasShreve}, Proposition
$5.2.13$ and Remark $5.3.3$, among others);
moreover, from (\ref{YamadaWatanabe}) and (\ref{YamadaWatanabe2}),
%
\begin{equation}
\label{LI} \int_{x-\varepsilon}^{x+\varepsilon}\frac{1 + |\mu(y)|}{\sigma^{2}(y)} \,dy
< +\infty\qquad \mbox{for some } \varepsilon>0,
\end{equation}
for every $x \in\inter(\mathcal{I})$.
Local integrability condition (\ref{LI}) implies that (\ref{SDE}) has a
weak solution (up to a possible explosion time) that is unique in the
sense of probability law (cf. \cite{KaratzasShreve}, Section $5.5$C).
Therefore, (\ref{SDE}) has a unique strong solution (possibly up to an
explosion time) due to \cite{KaratzasShreve}, Corollary $5.3.23$.
Also, it follows from (\ref{LI}) that the diffusion process $X^{x}$ is
regular in $\mathcal{I}$, that is, $X^{x}$ reaches $y$ with positive
probability starting at $x$, for any $x$ and $y$ in $\mathcal{I}$.
Hence, the state space $\mathcal{I}$ cannot be decomposed into smaller
sets from which $X^{x}$ could not exit (see, e.g., \cite{RogersWill},
Chapter VII).
We shall denote by $m(dx)$, $s(dx)$, $\mathcal{G}$ and $\mathbb{P}_x$
the speed measure, the scale function measure, the infinitesimal
generator and the probability measure such that $\mathbb{P}_x(\cdot) =
\mathbb{P}(\cdot| X(0)=x)$, $x \in\mathcal{I}$, respectively. Notice
that, under (\ref{LI}), $m(dx)$ and $s(dx)$ are well defined, and there
always exist two linearly independent, positive solutions of the
ordinary differential equation $\mathcal{G}u = \beta u$, $\beta> 0$
(cf. \cite{ItoMc}). These functions are uniquely defined up to
multiplication, if one of them is required to be strictly increasing
and the other to be strictly decreasing.
Finally, throughout this paper we assume that $\mathcal{I}$ is an
interval with endpoints $-\infty\leq\underline{x} < \overline{x}
\leq
+\infty$.

The firm's manager aims to increase the production capacity
%
\begin{equation}
\label{productioncapacity} C^{y,\nu}(t)= y + \nu(t),\qquad C^{y,\nu}(0)= y \geq0,
\end{equation}
by optimally choosing an irreversible investment plan $\nu\in\mathcal
{S}_o$, where
\begin{eqnarray*}
\mathcal{S}_o &:=& \bigl\{\nu\dvtx \Omega \times
\mathbb{R}_{+} \mapsto\mathbb{R}_{+}, \mbox{nondecreasing,
left-continuous, adapted}
\\
& &\hspace*{155pt} \mbox{such that } \nu(0)=0, \mathbb{P}\mbox{-a.s.}\bigr\}
\end{eqnarray*}
is the nonempty, convex set of irreversible investment processes.
The firm makes profit at rate $\pi(x,c)$ when its own capacity is $c$
and the status of the economy is~$x$, and the firm's manager discounts
revenues and costs at positive constant rate~$r$.
As for the operating profit function $\pi\dvtx  \mathcal{I} \times\mathbb
{R}_{+} \mapsto\mathbb{R}_{+}$, we make the following assumption.

%
\begin{Assumptions}
\label{AssProfit}
1. The mapping $c \mapsto\pi(x,c)$ is strictly increasing and
strictly concave with continuous derivative $\pi_{c}(x,c):=\frac
{\partial}{\partial c}\pi(x,c)$ on $\mathcal{I} \times(0,\infty)$
satisfying
\[
\lim_{c \rightarrow0}\pi_{c}(x,c)= \infty, \qquad\lim
_{c \rightarrow\infty}\pi_{c}(x,c)= \kappa,
\]
for some $0 \leq\kappa< \infty$.
\begin{longlist}[2.]
\item[2.] The process $(\omega,t) \mapsto\pi_c(X^{x}(\omega,t), y)$ is
$\mathbb{P} \otimes e^{-rt}\,dt$ integrable for any \mbox{$y > 0$}.
\end{longlist}
\end{Assumptions}

\begin{remark}
Notice that when $\kappa= 0$ we fall into the classical Inada
conditions which are satisfied, for example, by a Cobb--Douglas
operating profit.
In the case of a CES profit function of the form $\pi(x,c)=
(x^{{1}/{n}} + c^{{1}/{n}})^n$, $n\geq2$ (see Section~\ref{CESsubsection} below), one has instead $\kappa= 1$.
\end{remark}

The optimal investment problem is then
%
\begin{equation}
\label{valuefunction} V(x,y):=\sup_{\nu\in\mathcal{S}_o}\mathcal{J}_{x,y}(
\nu),
\end{equation}
where the profit functional $\mathcal{J}_{x,y}(\nu)$, net of investment
costs, is defined as
%
\begin{equation}
\label{netprofit} \mathcal{J}_{x,y}(\nu)=\mathbb{E} \biggl\{\int
_0^{\infty} e^{-r t} \pi
\bigl(X^{x}(t), C^{y,\nu}(t)\bigr) \,dt - \int
_0^{\infty} e^{- r t} \,d\nu(t) \biggr\}.
\end{equation}

Under Assumption~\ref{AssProfit}, $\mathcal{J}_{x,y}$ is well
defined but potentially infinite.
Since $\pi(x,\cdot)$ is strictly concave, $\mathcal{S}_o$ is convex and
$C^{y,\nu}$ is affine in $\nu$, then, if an optimal solution $\nu^{*}$
to (\ref{valuefunction}) does exist, it is unique.
Under further minor requirements the existence of a solution to (\ref
{valuefunction}) is a well-known result (see, e.g., \cite{RiedelSu},
Theorem~$2.3$, for an existence proof in a not necessarily Markovian framework).

\section{The optimal solution and the integral equation for the free-boundary}
\label{optimalsolutionFB}

A problem similar to (\ref{valuefunction}) (with depreciation in the
capacity dynamics) has been completely solved by Riedel and Su in \cite
{RiedelSu}, or (in the case of a time-dependent, stochastic finite
fuel) by Bank in \cite{Bank}. By means of a first-order conditions
approach and without relying on any Markovian or diffusive assumption,
these authors show that it is optimal to keep the production capacity
always above a desirable lower value of capacity, the \textsl{base
capacity} process (see \cite{RiedelSu}, Definition $3.1$), which is the
unique optional solution of a stochastic backward equation in the
spirit of Bank--El Karoui \cite{BankElKaroui}. In this section, we aim
to understand the meaning of the base capacity process $l^{*}$ in our setting.

Following \cite{Bank,CFR} or \cite{RiedelSu} (among others), we
start by deriving first-order conditions for optimality and by finding
the solution of (\ref{valuefunction}) in terms of a base capacity
process. Then, as a main new result, we identify the link between
$l^{*}$ and the free-boundary of the optimal stopping problem naturally
associated to the original singular control one (\ref{valuefunction})
and we determine an integral equation for the latter one.

Let $\mathcal{T}$ denote the set of all $(\mathcal{F}_t)$-stopping
times $\tau\geq0$ a.s. and notice that we may associate to $\mathcal
{J}_{x,y}(\nu)$ its supergradient as the unique optional process
defined by
%
\begin{equation}
\label{supergradient} \nabla\mathcal{J}_{x,y}(\nu) (\tau):= \mathbb{E} \biggl
\{\int_{\tau
}^{\infty} e^{- r s}
\pi_c\bigl(X^{x}(s), C^{y,\nu}(s)\bigr) \,ds \Big|
\mathcal {F}_{\tau} \biggr\} - e^{-r \tau},
\end{equation}
for any $\tau\in\mathcal{T}$.

\begin{remark}
\label{optionalsupergradient}
Following \cite{BankRiedel1}, Remark~3.1, among others, the quantity
$\nabla\mathcal{J}_{x,y}(\nu)(t)$ may be interpreted as the marginal
expected profit resulting from an additional infinitesimal investment
at time $t$ when the investment plan is $\nu$.
Mathematically, $\nabla\mathcal{J}_{x,y}(\nu)$ is the Riesz
representation of the profit gradient at $\nu$.
More precisely, define $\nabla\mathcal{J}_{x,y}(\nu)$ as the optional
projection of the product-measurable process
%
\begin{equation}
\label{progrmeas} \Phi(\omega,t):= \int_{t}^{\infty}
e^{-rs} \pi_{c}\bigl(X^{x}(\omega,s),C^{y,\nu}(\omega,s)\bigr) \,ds - e^{-rt},
\end{equation}
for $\omega\in\Omega$ and $t \geq0$. Hence, $\nabla\mathcal
{J}_{x,y}(\nu)$ is uniquely determined up to $\mathbb
{P}$-indistinguishability and it holds
\[
\mathbb{E} \biggl\{ \int_{0}^{\infty} \nabla
\mathcal{J}_{x,y}(\nu ) (t)\,d\nu (t) \biggr\} = \mathbb{E} \biggl\{ \int
_{0}^{\infty}\Phi(t) \,d\nu (t) \biggr\}
\]
for all admissible $\nu$ (cf. \cite{Jacod}, Theorem~1.33).
\end{remark}

\begin{theorem}
\label{FOCsthm}
Under Assumption~\ref{AssProfit}, a control $\nu^{*} \in\mathcal{S}_o$
is the unique optimal investment strategy for problem (\ref
{valuefunction}) if and only if the following first-order conditions
for optimality:
%
\begin{equation}
\label{FOCs} \cases{ %
\nabla
\mathcal{J}_{x,y}\bigl(\nu^{*}\bigr) (\tau) \leq0,&\quad $\mbox{a.s.
$\forall\tau\in\mathcal{T},$}$
\vspace*{2pt}\cr
\displaystyle\mathbb{E} \biggl\{\int_0^{\infty}\nabla\mathcal
{J}_{x,y}\bigl(\nu^{*}\bigr) (t) \,d\nu^{*}(t)
\biggr\} = 0, }
\end{equation}
hold true.
\end{theorem}

\begin{pf}
Sufficiency follows from concavity of $\pi(x,\cdot)$ (see, e.g.,
\cite
{Bank}),\break whereas for necessity see \cite{Steg}, Proposition $3.2$.
\end{pf}

Although the first-order conditions (\ref{FOCs}) completely
characterize the optimal investment plan $\nu^{*}$, they are not always
binding, and thus they cannot be directly applied to determine $\nu
^{*}$. Nevertheless, the optimal control may be obtained in terms of
the solution of a suitable Bank--El Karoui's representation problem
\cite
{BankElKaroui} related to (\ref{FOCs}).

For a fixed $T \leq+ \infty$, the Bank--El Karoui representation
theorem (cf. \cite{BankElKaroui}, Theorem~3 and Remark~2.1) states
that, given:
\begin{itemize}
\item an optional process $Y=\{Y(t), t \in[0,T]\}$ of class (D),
lower-semicontinuous in expectation with $Y(T)=0$,
\item a nonnegative, atomless optional random Borel measure $\mu
(\omega,dt)$ on $[0,T]$,
\item$f(\omega,t,x)\dvtx \Omega\times[0,T] \times\mathbb{R} \mapsto
\mathbb{R}$ such that $f(\omega, t, \cdot)\dvtx \mathbb{R} \mapsto
\mathbb
{R}$ is continuous, strictly decreasing from $+\infty$ to $-\infty$,
and the stochastic process $f(\cdot, \cdot,x)\dvtx\break \Omega\times[0,T]
\mapsto\mathbb{R}$ is progressively measurable and integrable with
respect to $d\mathbb{P} \otimes\mu(\omega,dt)$,
\end{itemize}
then there exists an optional process $\xi= \{\xi(t), t \in[0,T]\}$
taking values in $\mathbb{R} \cup\{-\infty\}$ such that for all
$\tau
\in\mathcal{T}$,
\[
f\Bigl(t,\sup_{\tau\leq u < t}\xi(u)\Bigr)\mathbh{1}_{(\tau, T]}(t)
\in L^1 \bigl(d\mathbb{P}\otimes\mu(\omega,dt) \bigr)
\]
and
%
\begin{equation}
\label{backwardgenerica} \mathbb{E} \biggl\{ \int_{(\tau, T]} f\Bigl(s, \sup
_{\tau\leq u < s} \xi (u)\Bigr) \mu(ds) \Big| \mathcal{F}_{\tau}
\biggr\}= Y(\tau).
\end{equation}
In \cite{BankElKaroui}, Lemma $4.1$ (see also \cite{BankFollmer},
Remark $1.4$(ii)), a real valued process $\xi$ is considered upper
right-continuous on $[0,T)$ if, for each $t$, $\xi(t) = \limsup_{s
\searrow t} \xi(s)$ with
%
\begin{equation}
\label{urc} \limsup_{s \searrow t} \xi(s):= \lim
_{\varepsilon\downarrow0} \sup_{s
\in[t, (t+\varepsilon) \wedge T]} \xi(s).
\end{equation}
Then, by \cite{BankElKaroui}, Theorem $1$, any progressively
measurable, upper right-continuous solution $\xi$ to (\ref
{backwardgenerica}) is uniquely determined up to optional sections on
$[0,T)$ in the sense that
\[
\xi(\tau) = \essinf_{\tau< \sigma\leq T}\Xi_{\tau,\sigma},\qquad \tau\in[0,T),
\]
where $\Xi_{\tau,\sigma}$ is the unique (up to a $\mathbb{P}$-null set)
$\mathcal{F}_{\tau}$-measurable random variable satisfying
\[
\mathbb{E}\bigl\{Y(\tau) - Y(\sigma) | \mathcal{F}_{\tau}\bigr\} =
\mathbb {E} \biggl\{\int_{(\tau, \sigma]} f(t,\Xi_{\tau,\sigma})
\mu(dt) \Big| \mathcal {F}_{\tau} \biggr\}.
\]

With $\kappa$ as in Assumption~\ref{AssProfit}, from now one we make
the following assumption.

\begin{Assumptions}
\label{rbiggerthankappa}
$r > \kappa$.
\end{Assumptions}

The following result holds.

\begin{proposition}
\label{existenceback}
Under Assumptions \ref{AssProfit} and \ref{rbiggerthankappa}, there
exists a unique (up to indistinguishability) strictly positive optional
solution $l^{*}$ to the backward stochastic equation
%
\begin{equation}
\label{backward} \mathbb{E} \biggl\{\int_{\tau}^{\infty}
e^{- r s} \pi_c\Bigl(X^{x}(s), \sup
_{\tau\leq u < s}l^{*}(u)\Bigr) \,ds \Big|\mathcal{F}_{\tau}
\biggr\} = e^{-r
\tau}, \qquad\tau\in\mathcal{T}.
\end{equation}
Moreover, the process $l^{*}$ has upper right-continuous paths.
\end{proposition}

\begin{pf}
Take $\kappa$ as in Assumption~\ref{AssProfit}, apply the Bank--El
Karoui representation theorem with $T=+\infty$ to
%
\begin{equation}
\label{identification} Y(\omega,t):= e^{-rt},\qquad \mu(\omega,dt):=
e^{- rt}\,dt
\end{equation}
and
%
\begin{equation}
\label{identificationf} f(\omega,t,l):= \cases{ %
\displaystyle\pi_c \biggl(X(\omega,t), -\frac{1}{l} \biggr),&\quad$\mbox{for }
l<0,$
\vspace*{2pt}\cr
-l + \kappa,&\quad$\mbox{for } l\geq0,$}
\end{equation}
and define
%
\begin{equation}
\label{definizionexil} \Xi^l(t):= \essinf_{\tau\geq t}\mathbb{E} \biggl\{
\int_t^{\tau
}f(s,l)\mu (ds) + Y(\tau) \Big|
\mathcal{F}_t \biggr\},\qquad l \in\mathbb{R}, t \geq0.
\end{equation}
Then, the optional process (cf. \cite{BankElKaroui}, equation (23)
and Lemma~4.13)
%
\begin{equation}
\label{defxi1} \xi^{*}(t) = \sup\bigl\{ l \in\mathbb{R}\dvtx
\Xi^l(t) = Y(t)\bigr\},\qquad t\geq0,
\end{equation}
solves the representation problem
%
\begin{equation}
\label{representationproblem0} \mathbb{E} \biggl\{ \int_{\tau}^{\infty}
e^{-rs} f\Bigl(s,\sup_{\tau
\leq u <
s} \xi^{*}(u)
\Bigr) \,ds\Big | \mathcal{F}_{\tau} \biggr\} = e^{-r\tau}, \qquad\tau\in
\mathcal{T}.
\end{equation}

If now $\xi^{*}$ has upper right-continuous paths and it is strictly
negative, then the strictly positive, upper right-continuous process
$l^{*}(t) = - \frac{1}{\xi^{*}(t)}$ solves
\begin{eqnarray*}
\label{representationproblem2} e^{-r\tau}& = &\mathbb{E} \biggl\{ \int
_{\tau}^{\infty} e^{- r s} \pi_{c}
\biggl(X^{x}(s), \frac{1}{-\sup_{\tau
\leq u < s}( - {1}/{(l^{*}(u))})} \biggr) \,ds \Big| \mathcal
{F}_{\tau
} \biggr\}
\\
&= & \mathbb{E} \biggl\{ \int_{\tau
}^{\infty}
e^{- r s} \pi_{c} \biggl(X^{x}(s),
\frac{1}{\inf_{\tau
\leq u
< s} ({1}/{(l^{*}(u))})} \biggr) \,ds \Big| \mathcal{F}_{\tau} \biggr\}
\\
& =&\mathbb{E} \biggl\{ \int_{\tau
}^{\infty}
e^{- r s} \pi_{c}\Bigl(X^{x}(s), \sup
_{\tau\leq u < s} l^{*}(u)\Bigr) \,ds \Big| \mathcal{F}_{\tau}
\biggr\},
\end{eqnarray*}
for any $\tau\in\mathcal{T}$, that is, $l^{*}$ solves (\ref
{backward}), thanks to (\ref{identificationf}) and (\ref
{representationproblem0}). Moreover, $\xi^{*}$ (and hence $l^{*}$) is
unique up to optional sections by \cite{BankElKaroui}, Theorem $1$, as
it is optional and upper right-continuous. Therefore, it is unique up
to indistinguishability by Meyer's optional section theorem (see, e.g.,
\cite{DellMeyer}, Theorem IV.86).

To complete the proof, we must show that $\xi^{*}(t)$ is indeed upper
right-continuous and strictly negative.
We start by proving its upper right-con\-tinuity. To accomplish that we
only need to prove that $\xi^{*}$ has upper semi-right-continuous
sample paths, that is,
%
\begin{equation}
\label{usrc} \limsup_{s \searrow t}\xi^{*}(s) \leq
\xi^{*}(t),
\end{equation}
since
\[
\limsup_{s \searrow t}\xi^{*}(s) \geq\xi^{*}(t)
\]
by definition [cf. (\ref{urc})].
Thanks to \cite{DellLeng}, Proposition $2$ (cf. also \cite
{BankKuchler}, proof of Theorem $1$) it suffices to show that $\lim_{n
\rightarrow\infty} \xi^{*}(\tau_n) \leq\xi^{*}(\tau)$, for any
sequence of stopping times $\{\tau_n\}_{n \geq1}$ such that $\tau_n
\downarrow\tau$ and for which there exists a.s. $\zeta:=\lim_{n
\rightarrow\infty}\xi^{*}(\tau_n)$.
Recall $\Xi^l$ of (\ref{definizionexil}), with $Y$, $\mu$ and $f$ as in
(\ref{identification}) and (\ref{identificationf}), and also that
$\xi^{*}(t) = \sup\{ l \in\mathbb{R} \dvtx \Xi^l(t) = Y(t)\}$ [cf.
\cite{BankElKaroui}, equation (23)].
Now, given $\varepsilon> 0$, for $\{\tau_n\}_{n \geq1}$ as above we have
\[
\Xi^{\zeta- \varepsilon}({\tau}) = \lim_{n \rightarrow\infty} \Xi ^{\zeta
- \varepsilon}({
\tau_n}) = Y(\tau),
\]
where we have used right-continuity of $t \mapsto\Xi^l(t)$, the fact
that $l \mapsto\Xi^l(t)$ is a continuous, decreasing mapping (cf.
\cite
{BankElKaroui}, Lemma $4.12$) and the threshold representation of $\xi^{*}$.
Hence, $\zeta- \varepsilon\leq\xi^{*}(\tau)$ and $\xi^{*}$ is upper
right-continuous because $\varepsilon>0$ was arbitrary.
Finally, we now show that $\xi^{*}$ is strictly negative. Define
\[
\sigma:= \inf\bigl\{t \geq0 \dvtx \xi^{*}(t) \geq0 \bigr\},
\]
then for $\omega\in\{\omega\dvtx  \sigma(\omega) < +\infty\}$, the upper
semi right-continuity of $\xi^{*}$ implies $\xi^{*}(\sigma)\geq0$, and
thus $\sup_{\sigma\leq u < s}\xi^{*}(u) \geq0$ for all $s > \sigma$.
Therefore, (\ref{representationproblem0}) with $\tau= \sigma$, that is,
%
\begin{equation}
\label{xipos} e^{-r \sigma} = \mathbb{E} \biggl\{ \int_{\sigma}^{\infty}
e^{- rs} \Bigl[-\sup_{\sigma\leq u < s}\xi^{*}(u) +
\kappa \Bigr] \,ds \Big| \mathcal{F}_{\sigma} \biggr\},
\end{equation}
or equivalently
\[
\biggl(\frac{r-\kappa}{r} \biggr)e^{-r \sigma} = - \mathbb{E} \biggl\{ \int
_{\sigma}^{\infty} e^{- rs} \sup
_{\sigma\leq u < s}\xi^{*}(u) \,ds \Big| \mathcal{F}_{\sigma}
\biggr\},
\]
is not possible for $\omega\in\{\omega\dvtx  \sigma(\omega) < +\infty\}$
since the right-hand side of (\ref{xipos}) is nonpositive, whereas the
left-hand side is always strictly positive due to Assumption~\ref
{rbiggerthankappa}. It follows that $\sigma=+\infty$ a.s., and hence
$\xi^{*}(t) < 0$ for all $t \geq0$ a.s.
\end{pf}

\begin{proposition}
\label{optimalsolthm}
Under Assumptions \ref{AssProfit} and \ref{rbiggerthankappa}, the
unique optimal irreversible investment process for problem (\ref
{valuefunction}) is given by
%
\begin{equation}
\label{optimalsol} \nu^{*}(t)=\Bigl(\sup_{0 \leq s < t}
l^{*}(s) - y \Bigr) \vee0,\qquad t>0, \nu^{*}(0)=0,
\end{equation}
where $l^{*}$ is the unique optional upper right-continuous solution to
(\ref{backward}).
\end{proposition}

\begin{pf}
See, for example, \cite{RiedelSu}, Theorem $3.2$.
\end{pf}

In the literature on stochastic, irreversible investment problems (cf.
\cite{KaratzasBaldursson,Chiarolla4,Chiarolla2} and
\cite{CF}, among others), or more generally on singular stochastic
control problems of monotone follower type
(see, e.g., \cite{Bank,KaratzasElKarouiSkorohod,KaratzasShreve84}), it is well
known that to a monotone control problem one may associate a suitable
optimal stopping problem whose optimal solution, $\tau^{*}$, is related
to the optimal control, $\nu^{*}$, by the simple relation $\tau
^{*}=\inf
\{t \geq0\dvtx \nu^{*}(t) > 0\}$.
Economically, it means that a firm's manager has to decide how to
optimally invest or, equivalently, when to profitably exercise the
investment option.
Indeed, if we introduce for any $\nu\in\mathcal{S}_o$ the level
passage times $\tau^{\nu}(q):=\inf\{t \geq0\dvtx \nu(t) > q\}$, $q
\geq
0$, then for every $x \in\mathcal{I}$ and $y \geq0$ we may write
(cf., e.g., \cite{KaratzasBaldursson}, Lemma $2$)
\begin{eqnarray*}
\mathcal{J}_{x,y}(\nu) - \mathcal{J}_{x,y}(0)& = & \int
_y^{\infty}\mathbb{E} \biggl\{\int
_{\tau
^{\nu
}(z-y)}^{\infty} e^{-rs}\pi_c
\bigl(X^{x}(s), z\bigr) \,ds - e^{-r \tau^{\nu
}(z-y)} \biggr\} \,dz
\\
& \leq& \int_y^{\infty} \sup_{\tau
\geq0}
\mathbb{E} \biggl\{\int_{\tau}^{\infty} e^{-rs}
\pi _c\bigl(X^{x}(s), z\bigr) \,ds - e^{-r \tau}
\biggr\} \,dz
\\
& = & \int_y^{\infty} \mathbb {E} \biggl\{ \int
_{0}^{\infty} e^{-rs}\pi_c
\bigl(X^{x}(s), z\bigr) \,ds \biggr\} \,dz
\\
& & {}- \int_y^{\infty} \inf_{\tau\geq0}
\mathbb {E} \biggl\{\int_{0}^{\tau}
e^{-rs}\pi_c\bigl(X^{x}(s), z\bigr) \,ds +
e^{-r \tau} \biggr\} \,dz.
\end{eqnarray*}
Therefore, if a process $\nu^{*} \in\mathcal{S}_o$ is such that its
level passage times are optimal for the previous optimal stopping
problems, then $\nu^{*}$ must be optimal for problem (\ref{valuefunction}).
Hence,
%
\begin{equation}
\label{v} v(x,y):=\inf_{\tau\geq0}\mathbb{E} \biggl\{\int
_0^{\tau} e^{-rs}\pi _c
\bigl(X^{x}(s),y\bigr) \,ds + e^{-r \tau} \biggr\}
\end{equation}
is the optimal timing problem naturally associated to the optimal
investment problem (\ref{valuefunction}).
Notice that $v(x,y) \leq1$, for all $x \in\mathcal{I}$ and $y > 0$,
and that the mapping $y \mapsto v(x,y)$ is nonincreasing for any $x \in
\mathcal{I}$, because $\pi(x,\cdot)$ is strictly concave.
We may now define the continuation region
%
\begin{equation}
\label{continuation} \mathcal{C}:=\bigl\{(x,y) \in\mathcal{I} \times(0,\infty)\dvtx
v(x,y) < 1\bigr\}
\end{equation}
and the stopping region
%
\begin{equation}
\label{stopping} \mathcal{S}:= \bigl\{(x,y) \in\mathcal{I} \times(0,\infty)\dvtx
v(x,y) = 1\bigr\}.
\end{equation}
Intuitively, $\mathcal{S}$ is the region in which it is optimal to
invest immediately, whereas $\mathcal{C}$ is the region in which it is
profitable to delay the investment option. The nonincreasing property
of $y \mapsto v(x,y)$ implies that $\mathcal{S}$ is below $\mathcal{C}$
and, therefore, that
%
\begin{equation}
\label{boundary} b(x):= \sup\bigl\{y > 0\dvtx v(x,y)=1\bigr\},\qquad x \in\mathcal{I},
\end{equation}
is the boundary between these two regions, that is, the free-boundary.

%
\begin{Assumptions}
\label{xpicnondecr}
The mapping $x \mapsto\pi_c(x,c)$ is nondecreasing for any $c \in
(0,\infty)$.
\end{Assumptions}

Notice that, if $\pi$ were twice continuously differentiable, then
Assumption~\ref{xpicnondecr} would mean that $\pi$ is supermodular.
In \cite{RiedelSu}, Section $5$, supermodularity of the profit function
has been used to derive comparative statics results for the base
capacity process $l^{*}$. It is easy to see that Cobb--Douglas and CES
profit functions are supermodular on $(0,\infty) \times(0, \infty)$.
Condition \ref{xpicnondecr} has also a reasonable economic meaning (see
also the discussion in \cite{MehriZervos}, page~844, in the context of
a stochastic, reversible investment problem). Indeed, if the process
$X$ models the uncertain status of the market as, for example, the
price of or the demand for the produced good, then it seems natural to
imagine that marginal profits are positively affected by improving
market conditions.

\begin{proposition}
\label{vincreasingprop}
Under Assumptions \ref{AssProfit} and \ref{xpicnondecr}, $x \mapsto
v(x,y)$ is nondecreasing for any $y >0$.
\end{proposition}

\begin{pf}
For $y>0$, take $x_1 > x_2$, $x_1,x_2 \in\mathcal{I}$, let $\tau^{*}
\in\mathcal{T}$ be optimal for $(x_1,y)$ and $\theta\in\mathcal{T}$
be a generic stopping time. Then
\begin{eqnarray*}
\hspace*{-4pt}& & v(x_1,y) - v(x_2,y) \\
\hspace*{-4pt}&&\qquad\geq\mathbb{E} \biggl\{\int
_0^{\tau^{*}} e^{-rs} \pi_c
\bigl(X^{0,x_1}(s), y\bigr)\,ds + e^{-r\tau^{*}} - \int_0^{\theta} \pi_c
\bigl(X^{0,x_2}(s), y\bigr) \,ds - e^{-r\theta} \biggr\},
\end{eqnarray*}
for any $\theta\in\mathcal{T}$. Take now $\theta\equiv\tau^{*}$
to obtain
\[
v(x_1,y) - v(x_2,y) \geq\mathbb{E} \biggl\{\int
_0^{\tau^{*}} e^{-rs} \bigl[\pi
_c\bigl(X^{0,x_1}(s), y\bigr) - \pi_c
\bigl(X^{0,x_2}(s), y\bigr)\bigr] \,ds \biggr\} \geq0,
\]
since $x \mapsto X^{x}(t)$ is a.s. increasing for any $t \geq0$ due
to the Yamada--Watanabe comparison theorem (see, e.g., \cite
{KaratzasShreve}, Propositions $5.2.13$ and $5.2.18$) thanks to our
conditions (\ref{YamadaWatanabe}) and (\ref{YamadaWatanabe2}).
\end{pf}

\begin{corollary}
\label{bincreasingprop}
Let Assumptions \ref{AssProfit} and \ref{xpicnondecr} hold. Then, the
free-boundary $b(\cdot)$ between the continuation region and the
stopping region is nondecreasing for any $x \in\mathcal{I}$.
\end{corollary}

\begin{pf}
Use the result of Proposition~\ref{vincreasingprop} and arguments
similar to those in~\cite{Jacka}, proof of Proposition $2.2$.
\end{pf}

The next theorem gives us a new representation for the base capacity
$l^{*}$ in our setting.

%
\begin{theorem}
\label{identificocor}
Let $l^{*}$ be the unique optional solution of (\ref{backward}) and
$b(\cdot)$ the free-boundary defined in (\ref{boundary}).
Under Assumptions \ref{AssProfit}, \ref{rbiggerthankappa} and \ref
{xpicnondecr}, one has
%
\begin{equation}
\label{identifico} l^{*}(t) = b\bigl(X^{x}(t)\bigr).
\end{equation}
\end{theorem}

\begin{pf}
First of all notice that the right-hand side of (\ref{identifico}) is
an optional process as well as $l^{*}$, being $b(\cdot)$ a
Borel-measurable function (since monotone) and $X$ optional. To prove
(\ref{identifico}) recall that $l^{*}(t) = - \frac{1}{\xi^{*}(t)}$
(cf. proof of Proposition~\ref{existenceback}) and that the process
$\xi^{*}$ admits the representation (cf. \cite{BankElKaroui}, formula
$(23)$ on page $1049$)
\[
\label{rappresentoxistar} \xi^{*}(t) = \sup \biggl\{ l < 0 \dvtx
\essinf_{\tau\geq t}\mathbb {E} \biggl\{ \int_t^{\tau}e^{-rs}
\pi_c\biggl(X^{x}(s), -\frac{1}{l}\biggr) \,ds +
e^{-r\tau
} \Big|\mathcal{F}_t \biggr\} = e^{-rt} \biggr
\}.
\]

To take care of the previous conditional expectation, we adapt the
arguments of~\cite{Chiarolla4}, proof of Theorem $4.1$. Let
$(\overline
{\Omega}, \overline{\mathbb{P}})$ be the canonical probability space
where $\overline{\mathbb{P}}$ is the Wiener measure on $\overline
{\Omega
}:=\mathcal{C}_0([0,\infty); \mathbb{R}^2)$, the space of all
continuous functions from $[0,\infty)$ to $\mathbb{R}^2$ which are zero
at $t=0$. We denote by $W(t,\overline{\omega})=\overline{\omega}(t)$
the coordinate mapping on $\mathcal{C}_0([0,\infty); \mathbb{R}^2)$,
with $\overline{\omega}(t):= (\overline{\omega}_1,\overline{\omega
}_2)(t)$, $\overline{\omega}_1:=\{W(u), 0 \leq u \leq t\}$ and
$\overline{\omega}_2:=\{W(u)-W(t), u \geq t\}= \{W'(u), u\geq0\}$.
Independence of Brownian increments induces a product measure on
$\mathcal{C}_0([0,\infty); \mathbb{R}^2)=C_{0}([0,t]; \mathbb
{R})\times
C_{0}([t,\infty); \mathbb{R})$. Then $\tau(\overline{\omega
}_1,\overline
{\omega}_2)= t + \tau'_{\overline{\omega}_1}(\overline{\omega}_2)$
[where for each $\overline{\omega}_1$, $\tau'_{\overline{\omega
}_1}(\cdot)$ is a stopping time with respect to $\{\mathcal
{F}^{W'}_u\}
_{u \geq0}$] and we may write
\begin{eqnarray*}
\label{takecare} & & \mathbb{E} \biggl\{\int_t^{\tau}e^{-rs}
\pi_c\biggl(X^{x}(s), -\frac
{1}{l}\biggr)\,ds +
e^{-r\tau} \Big|\mathcal{F}_t \biggr\}
\\
& & \qquad= e^{-rt} \mathbb{E} \biggl\{\int_0^{\tau'_{\overline{\omega}_1}}
e^{-r u}\pi_c\biggl(X^{t,X^{x}(t)}(u + t), -
\frac{1}{l}\biggr) \,du + e^{-r \tau
'_{\overline{\omega}_1}} \Big|\mathcal{F}_t \biggr
\}
\\
& &\qquad = e^{-rt} \mathbb{E}_{\overline{\omega}_2} \biggl\{\int
_0^{\tau
'_{\overline{\omega}_1}} e^{-r u}\pi_c
\biggl(X^{t,X^{x}(t)}(u+t), -\frac
{1}{l}\biggr) \,du + e^{-r \tau'_{\overline{\omega}_1}}
\biggr\}
\\
& & \qquad= e^{-rt} \Psi\bigl(X^{x}(t); \tau'_{\overline{\omega}_1}
\end{eqnarray*}
for any $\overline{\omega}_1$ fixed, for some $\Psi$ and where
$\mathbb
{E}_{\overline{\omega}_2}\{\cdot\}$ denotes the expectation over
$\overline{\omega}_2$ or~$W'$. But now $X$ is a time-homogeneous
diffusion, hence
\[
\Psi\bigl(z;\tau'_{\overline{\omega}_1}\bigr)= \mathbb{E} \biggl\{\int
_0^{\tau
'_{\overline{\omega}_1}} e^{-r u}\pi_c
\biggl(X^{0,z}(u), -\frac{1}{l}\biggr) \,du + e^{-r \tau'_{\overline{\omega}_1}} \biggr
\},
\]
for any $\overline{\omega}_1$ given and fixed, and thus
\begin{eqnarray*}
\label{reprxistar} \xi^{*}(t)& = & \sup \biggl\{ l < 0 \dvtx
\essinf_{\tau\geq t}\mathbb{E} \biggl\{\int_t^{\tau}e^{-rs}
\pi _c\biggl(X^{x}(s), -\frac{1}{l}\biggr) \,ds +
e^{-r\tau}\Big |\mathcal{F}_t \biggr\} = e^{-rt} \biggr
\}
\\
& = & \sup \biggl\{ l < 0 \dvtx v\biggl(X^{x}(t), -\frac{1}{l}
\biggr) = 1 \biggr\},
\nonumber
\end{eqnarray*}
with $v$ as in (\ref{v}).

Finally, since $l^{*}(t) = - \frac{1}{\xi^{*}(t)}$ (cf. proof of
Proposition~\ref{existenceback}), we may write for $y>0$
\begin{eqnarray*}
\label{reprellestar} l^{*}(t)& = & - \frac{1}{\sup
\{ l
< 0 \dvtx v(X^{x}(t), -{1}/{l}) = 1 \}}
\\
& = & \frac{1}{-\sup \{ -
{1}/{y} < 0 \dvtx v(X^{x}(t), y) = 1 \}}
\\
& = & \frac{1}{\inf \{ {1}/{y}
> 0 \dvtx v(X^{x}(t), y) = 1 \}}
\\
& = & \sup \bigl\{ y > 0 \dvtx v\bigl(X^{x}(t), y\bigr) = 1 \bigr\},
\end{eqnarray*}
and then the thesis follows by (\ref{boundary}).
\end{pf}

\begin{remark}
The result of Theorem~\ref{identificocor} still holds if one introduces
depreciation in the production capacity dynamics as in \cite{RiedelSu};
that is, if
\[
C^{y,\nu}(t)=-\rho C^{y,\nu}(t) \,dt + d\nu(t),\qquad C^{y,\nu}(0)=y
\geq0,
\]
for some $\rho>0$.
Moreover, in this case, one has (cf. also \cite{RiedelSu}, Theorem~3.2)
\[
\nu^{*}(t)=\int_{[0,t)}e^{-\rho s} \,d\overline{
\nu}^{*}(s)\qquad\mbox{with } \overline{\nu}^{*}(t)= \sup
_{0 \leq s < t} \biggl(\frac
{b(X^{x}(s)) - ye^{-\rho s}}{e^{-\rho s}} \biggr) \vee0
\]
and $\overline{\nu}^{*}(0)=0$.
\end{remark}

Theorem~\ref{identificocor} clarifies why in the literature (cf. \cite
{BankRiedel1,CFR} or \cite{RiedelSu}, among others) one usually
refers to $l^{*}$ as a ``desirable value of capacity'' that the
controller aims to maintain in a ``minimal way.'' Indeed, as in the
classical monotone follower problems (see, e.g., \cite
{KaratzasElKarouiSkorohod} and \cite{KaratzasShreve84}), the optimal
investment policy $\nu^{*}$ (cf. Proposition~\ref{optimalsolthm}) is
the solution of a Skorohod problem being the least effort needed to
reflect the production capacity at the moving (random) boundary
$l^{*}(t)=b(X^{x}(t))$, that is,
\[
\nu^{*}(t)= \sup_{0 \leq s < t}\bigl(b\bigl(X^{x}(s)
\bigr) - y\bigr) \vee0,\qquad t>0, \nu^{*}(0)=0.
\]

The result of Theorem~\ref{identificocor} resembles those of \cite
{BankFollmer} and \cite{BankBaumgarten} in which the connection between
the solution of a Bank--El Karoui representation problem and a suitable
exercising boundary for parameter-dependent optimal stopping problems
has been pointed out.
In particular, in \cite{BankBaumgarten} the authors consider the
optimal stopping problem $\sup_{\tau\geq0}\mathbb{E}\{e^{-r\tau
}(u(X^x(\tau))-k)\}$ where $X$ is a regular, one-dimensional diffusion
and $k$ a real parameter which affects linearly the gain function.
Under some additional uniform integrability conditions on $X$, they can
show that the solution $K$ of an associated representation problem is
given by $K(t)=\gamma(X(t))$, with $\gamma(\cdot)$ the free-boundary on
the $(x,k)$-plane (see also \cite{ElKarouiFollmer}, Sections $4$ and
$5$). Moreover, $\gamma(\cdot)$ is characterized in terms of the
infimum of an auxiliary function of one variable that can be determined
from the Laplace transforms of level passage times for~$X$.

When our marginal profit $\pi_c$ is multiplicatively separable [i.e.,
$\pi_c(x,c)=f(x)g(c)$, as in the Cobb--Douglas case], it is not hard to
see that our optimal stopping problem (\ref{v}) may be reduced to that
studied in \cite{BankBaumgarten} (set $u(x):=\mathbb{E}\{\int_0^{\infty
}e^{-rs}f(X^{x}(s))\,ds\}$ and $k:=1/g(y)$ to obtain by the strong Markov
property $v(x,y)= g(y)[ u(x) - \sup_{\tau\geq0}\mathbb{E}\{
e^{-r\tau
}(u(X^x(\tau))-k)\}]$). However, we shall start from the identification
(\ref{identifico}) to find, by (\ref{backward}) and by purely
probabilistic arguments, an integral equation for the free-boundary
(cf. Theorem~\ref{integraleqthm} below) which holds for a very general
class of concave profit functions not necessarily multiplicatively
separable. That is, for example, the case of a CES (constant elasticity
of substitution) profit that we will discuss in Section~\ref{Examples}.

\begin{theorem}
\label{integraleqthm}
Let Assumptions \ref{AssProfit}, \ref{rbiggerthankappa} and \ref
{xpicnondecr} hold. Denote by $\mathcal{G}$ the infinitesimal generator
associated to $X^{x}$, and by $\psi_r(x)$ the increasing solution to
the equation $\mathcal{G}u = ru$. Moreover, let $m(dx)$ and $s(dx)$ be
the speed measure and the scale function measure, respectively,
associated to the diffusion $X^{x}$. Then the free-boundary $b(\cdot)$
between the continuation region and the stopping region is the unique
positive nondecreasing solution to the integral equation
%
\begin{equation}
\label{integraleq} \psi_r(x)\int_{x}^{\overline{x}}
\biggl(\int_{\underline{x}}^z \pi _c
\bigl(y,b(z)\bigr) \psi_r(y) m(dy) \biggr) \frac{s(dz)}{\psi_r^2(z)} = 1.
\end{equation}
\end{theorem}

\begin{pf}
Since $l^{*}$ uniquely solves (\ref{backward}) and $l^{*}(t) =
b(X^{x}(t))$ (cf. Theorem~\ref{identificocor}), then $b(\cdot)$ satisfies
%
\begin{eqnarray}
\label{eqint1} r & = & \mathbb{E} \biggl\{\int_{\tau
}^{\infty}r
e^{-r(s-\tau)}\pi_c\Bigl(X^{x}(s), \sup
_{\tau\leq u <
s}b\bigl(X^{x}(u)\bigr)\Bigr) \,ds\Big |
\mathcal{F}_{\tau} \biggr\}
\nonumber
\\[-8pt]
\\[-8pt]
\nonumber
& = & \mathbb{E} \biggl\{\int_{0}^{\infty
}r
e^{-rt}\pi_c\Bigl(X^{x}(t + \tau), b\Bigl(\sup
_{0 \leq u < t}X^{x}(u + \tau)\Bigr)\Bigr) \,dt \Big|
\mathcal{F}_{\tau} \biggr\},
\end{eqnarray}
for any $\tau\in\mathcal{T}$, where in the second equality we have
used the fact that $b(\cdot)$ is nondecreasing by Corollary~\ref
{bincreasingprop}.
Now, by the strong Markov property, (\ref{eqint1}) amounts to find
$b(\cdot)$ such that
\[
\label{eqint2} \mathbb{E}_x \biggl\{\int_{0}^{\infty}r
e^{-rt}\pi_c\Bigl(X(t), b\Bigl(\sup_{0 \leq
u < t}X(u)
\Bigr)\Bigr) \,dt \biggr\} = r;
\]
that is, such that
\[
\label{eqint3} \mathbb{E}_x \bigl\{\pi_c\bigl(X(
\tau_r), b\bigl(M(\tau_r)\bigr)\bigr) \bigr\} = r,
\]
where $M(t):=\sup_{0 \leq s \leq t}X(s)$ and $\tau_r$ denotes an
independent exponentially distributed random time with parameter $r$.
Integral equation (\ref{integraleq}) now follows since for a
one-dimensional regular diffusion $X$ (cf. \cite{Borodin}, page $185$)
one has
\[
\mathbb{P}_x\bigl(X(\tau_r) \in dy, M(
\tau_r) \in dz\bigr) = r\frac{\psi
_r(x)\psi
_r(y)}{\psi_r^2(z)}m(dy)s(dz),\qquad y\leq z, x
\leq z.
\]

Finally, uniqueness of a positive, nondecreasing $b(\cdot)$ satisfying
(\ref{integraleq}) can be proved arguing by contradiction as follows.
Assume there exist two positive, nondecreasing solutions $b_1(\cdot)$
and $b_2(\cdot)$ of (\ref{integraleq}) such that $b_1(x_o) \neq
b_2(x_o)$ for some $x_o \in\mathcal{I}$. Then, proceeding backward
from (\ref{integraleq}), one finds two positive, optional processes
$l^{*}_1(t):=b_1(X^x(t))$ and $l^{*}_2(t):=b_2(X^x(t))$ both solving
(\ref{backward}). By Proposition~\ref{existenceback}, we should have
$l^{*}_1$ and $l^{*}_2$ indistinguishable. But now $X$ is regular, and
thus the set $\{\omega\in\Omega\dvtx  \tau_{x_o}(\omega) < +\infty\}$,
with $\tau_{x_o}:=\inf\{t \geq0\dvtx X^x(t)=x_o\}$, has positive
probability for any $x$ in the interior of $\mathcal{I}$. It follows
that $l^{*}_1$ and $l^{*}_2$ are not indistinguishable and such a
contradiction completes the proof.
\end{pf}

Notice that if one deals with an optimal stopping problem of type (\ref
{v}), the common approach consists in writing down the associated
free-boundary problem for the value function $v$ and the boundary $b$
and try to solve it on a case by case basis.
Alternatively, one could rely on an integral representation for the
value function and the free-boundary which follows from the local
time--space calculus for semimartingales on continuous surfaces of
Peskir \cite{Peskir}. The latter, indeed, may be seen as the
probabilistic counterpart of the free-boundary problem. However, for
both of these two approaches one needs regularity of $v$, smooth-fit
property or a priori continuity of $b$.

Our integral equation (\ref{integraleq}), instead, follows immediately
from the backward equation (\ref{backward}) for $l^{*}(t)=b(X^{x}(t))$,
thanks to (\ref{identifico}) and the strong Markov property of $X$.
Therefore, it does not require any regularity of the value function,
smooth-fit property or a priori continuity of $b(\cdot)$ itself to be
applied. It thus represents an extremely useful tool to determine the
free-boundary of the whole class of infinite time horizon, singular
stochastic irreversible investment problems of type (\ref
{valuefunction}). As we shall see in the next section, equation (\ref
{integraleq}) may be analytically solved even in some nontrivial cases.\looseness=-1

\section{Explicit results}
\label{Examples}

In this section, we aim to explicitly solve the integral equation (\ref
{integraleq}) when the economic shock $X^{x}$ is a geometric Brownian
motion, a~three-dimensional Bessel process and a CEV (constant
elasticity of volatility) process. We shall find the free-boundary
$b(\cdot)$ of the optimal stopping problem (\ref{v}) for Cobb--Douglas
and CES (constant elasticity of substitution) operating profit
functions, that is, for $\pi(x,c) = \frac{x^{\alpha}c^{\beta
}}{\alpha+
\beta}$ with $\alpha, \beta\in(0,1)$, and $\pi(x,c)= (x^{{1}/{n}}
+ c^{{1}/{n}})^n$, $n\geq2$, respectively.

To the best of our knowledge, this is the first time that the
free-boundary of a singular stochastic control problem of type (\ref
{valuefunction}) with a CES profit function is explicitly determined
for underlying given by a three-dimensional Bessel process or by a CEV process.

\subsection{The case of a Cobb--Douglas operating profit}

Throughout this section, assume that the operating profit function is
of Cobb--Douglas type, that is, $\pi(x,c) = \frac{x^{\alpha}c^{\beta
}}{\alpha+ \beta}$ for $\alpha, \beta\in(0,1)$. According to
Assumption~\ref{rbiggerthankappa}, we take $r>0$.

\subsubsection{Geometric Brownian motion}
\label{GBMCD}
Let $X^{x}(t) = xe^{(\mu-({1}/{2})\sigma^2)t + \sigma W(t)}$,\break $x >0$,
with $\sigma^2 > 0$ and $\mu\in\mathbb{R}$.
If we denote by $\delta:=\frac{\mu}{\sigma^2} - \frac{1}{2}$, then it
is well known (cf., e.g., \cite{BorSal}) that
\[
m(dx)=\frac{2}{\sigma^2}x^{2\delta-1}\,dx
\]
and
\[
\label{scaleBG} s(dx):= \cases{ %
x^{-2\delta-1}\,dx,&\quad $\delta\neq0,$
\vspace*{2pt}\cr
\displaystyle\frac{1}{x}\,dx, &\quad $\delta=0.$ }
\]
Finally, the ordinary differential equation $\mathcal{G}u = ru$, that
is, $\frac{1}{2}\sigma^2 x^2 u''(x) + \mu x u'(x) = ru$, admits the
increasing solution
\[
\psi_r(x)= x^{\gamma_1},
\]
where $\gamma_1$ is the positive root of the equation $\frac
{1}{2}\sigma
^2\gamma(\gamma-1) + \mu\gamma= r$.
%

\begin{proposition}
For any $\delta\in\mathbb{R}$ and $x > 0$, one has
%
\begin{equation}
\label{bGBM} b(x) = K_{\delta} x^{{\alpha}/{(1-\beta)}},
\end{equation}
with $K_{\delta}:=  [\sigma^2\gamma_1(\alpha+ \gamma_1 +
2\delta
) (\frac{\alpha+ \beta}{2\beta} ) ]^{-{1}/{(1
-\beta)}}$.
\end{proposition}

\begin{pf}
Let us start with the case $\delta\neq0$. For any $x>0$ by (\ref
{integraleq}), we have
\[
\int_x^{\infty} \biggl(\int_0^z
y^{\alpha+ \gamma_1 + 2\delta-1} \,dy \biggr) b^{\beta-1}(z) z^{-2\delta- 1 -2\gamma_1} \,dz =
x^{-\gamma
_1} \biggl(\frac{\alpha+ \beta}{2\beta} \biggr)\sigma^2;
\]
that is,
\[
\int_x^{\infty} b^{\beta-1}(z) z^{\alpha- \gamma-1}
\,dz = \sigma ^2(\alpha+ \gamma_1 + 2\delta) \biggl(
\frac{\alpha+ \beta}{2\beta
} \biggr)x^{-\gamma_1}.
\]
Take now $b(z) = (A_{\delta} z)^{{\alpha}/{(1-\beta)}}$, for some
constant $A_{\delta}$, to obtain
\[
A_{\delta}^{-\alpha} \int_x^{\infty}
z^{-\gamma_1-1} \,dz = \frac
{A_{\delta}^{-\alpha}}{\gamma_1}x^{-\gamma_1} = \sigma^2(
\alpha+ \gamma _1 + 2\delta) \biggl(\frac{\alpha+ \beta}{2\beta}
\biggr)x^{-\gamma_1},
\]
which is satisfied by $A_{\delta}:=  [\sigma^2\gamma_1(\alpha+
\gamma_1 + 2\delta) (\frac{\alpha+ \beta}{2\beta}
)
]^{-{1}/{\alpha}}$. Hence,
$b(x)= K_{\delta}x^{{\alpha}/{(1-\beta)}}$ with $K_{\delta
}:=A_{\delta
}^{{\alpha}/{(1-\beta)}}$.
Similar calculations also apply to the case $\delta=0$ to have $b(x) =
K_{0} x^{{\alpha}/{(1-\beta)}}$.
\end{pf}

\subsubsection{Three-dimensional Bessel process}
\label{Bessel2}

Let now $X^{x}$ be a three-dimen\-sional Bessel process, that is, the
strong solution of
\[
dX^{x}(t) = \frac{1}{X^{x}(t)}\,dt + dW(t),\qquad X^{x}(0)=x>0.
\]
It is a diffusion with state space $(0,\infty)$, generator $\mathcal
{G}:=\frac{1}{2}\frac{d}{dx^2} + \frac{1}{x}\frac{d}{dx}$ and scale and
speed measures given by $s(dx) = x^{-2}\,dx$ and $m(dx)=2x^2 \,dx$,
respectively (cf.~\cite{JYC}, Chapter~VI). Further, since $X^{x}(t)$
may be characterized as a killed Brownian motion at zero, conditioned
never to hit zero, the three-dimensional Bessel process may be viewed
as an excessive transform of a killed Brownian motion with excessive
function $h(x)=x$, that is, the scale function of the Brownian motion.
Therefore, $\psi_r(x) = \frac{\sinh{(\sqrt{2r}x)}}{x}$ (cf. \cite{JYC},
Chapter VI or \cite{Borodin}, Section $3.2$, among others).

\begin{remark}
Notice that the first of (\ref{YamadaWatanabe}) is not satisfied in
this case. However, that condition is not necessary to obtain our
Theorem~\ref{integraleqthm}. In fact, we only need that $X$ is a
diffusion process for which a comparison result holds true. One can see
that this fact is satisfied by a three-dimensional Bessel process since
it is the squared root of a squared Bessel process for which a
comparison result (cf. the Yamada--Watanabe theorem, e.g., \cite
{KaratzasShreve}, Propositions 5.2.13 and 5.2.18) holds.
\end{remark}

The following result holds.

%
\begin{proposition}
For any $x>0$, one has
%
\begin{equation}
\label{FBBesselCD} b(x) = \biggl[ \biggl(\frac{\alpha+ \beta}{2\beta} \biggr) x^2
\frac
{\psi
'_r(x)}{g(x)} \biggr]^{-{1}/{(1-\beta)}},
\end{equation}
where $\psi'_r(x)$ denotes the first derivative of the increasing
function $\psi_r(x) = \frac{\sinh{(\sqrt{2r}x)}}{x}$, and
$g(x):=\int_0^x y^{\alpha+1}\sinh{(\sqrt{2r}y)}\,dy$.
\end{proposition}

\begin{pf}
From integral equation (\ref{integraleq}), we may write
\begin{eqnarray*}
\label{BesselCD1} \biggl(\frac{\alpha+ \beta}{2\beta} \biggr)
\frac{x}{\sinh{(\sqrt {2r}x)}} & = & \int
_x^{\infty
} \biggl(\int_0^z
y^{\alpha+1}\sinh{(\sqrt{2r}y)}\,dy \biggr) \frac{b^{\beta
-1}(z)}{\sinh^2{(\sqrt{2r}z)}} \,dz
\nonumber
\\
& = & \int_x^{\infty} g(z) \frac
{b^{\beta-1}(z)}{\sinh^2{(\sqrt{2r}z)}} \,dz,
\end{eqnarray*}
with $g(x):= \int_0^x y^{\alpha+1}\sinh{(\sqrt{2r}y)}\,dy$. By
differentiating, one obtains
%
\begin{eqnarray}
\label{BesselCD2} b^{\beta-1}(x) & = & \biggl(\frac
{\alpha+ \beta}{2\beta} \biggr)
\frac{[x\sqrt{2r}\cosh{(\sqrt {2r}x)} -
\sinh{(\sqrt{2r}x)}]}{g(x)}
\nonumber
\\[-8pt]
\\[-8pt]
\nonumber
& = & \biggl(\frac{\alpha+ \beta
}{2\beta} \biggr) x^2 \frac{\psi'_r(x)}{g(x)},
\end{eqnarray}
that is
\[
b(x)= \biggl[ \biggl(\frac{\alpha+ \beta}{2\beta} \biggr) x^2
\frac
{\psi
'_r(x)}{g(x)} \biggr]^{-{1}/{(1-\beta)}}.
\]
Notice that $b(\cdot)$ is positive since $\psi_r(\cdot)$ is increasing
and $g(\cdot)$ is positive.

To complete the proof, it now suffices to check that the mapping $x
\mapsto b(x)$ is actually nondecreasing as suggested by Proposition~\ref
{bincreasingprop}, that is, $x \mapsto b^{\beta-1}(x)$ is nonincreasing.
From (\ref{BesselCD2}), we have
%
\begin{eqnarray}
\label{derivataFBCD}\qquad \frac{d}{dx}b^{\beta-1}(x) & = & \biggl(
\frac{\alpha+ \beta}{2\beta g^2(x)} \biggr) \bigl[g(x) \bigl(2x\psi'_r(x)
+ x^2\psi^{\prime\prime}_r(x)\bigr) -
g'(x)x^2\psi'_r(x) \bigr]
\nonumber
\\[-8pt]
\\[-8pt]
\nonumber
& = & \biggl(\frac{x^2(\alpha+
\beta
)}{2\beta g^2(x)} \biggr) \bigl[2rg(x) \psi_r(x) -
g'(x) \psi '_r(x) \bigr],
\end{eqnarray}
since $\psi_r(x)$ solves $\frac{1}{2}\psi^{\prime\prime}_r(x) + \frac
{1}{x}\psi
'_r(x) = r \psi_r(x)$.
Recall now that $\psi_r(x)=\break \frac{\sinh{(\sqrt{2r}x)}}{x}$,
$g'(x)=x^{\alpha+1}\sinh{(\sqrt{2r}x)}$ and notice that, by an
integration by parts,
\[
g(x)= \int_0^x y^{\alpha+1}\sinh{(
\sqrt{2r}y)}\,dy = \frac{1}{\sqrt {2r}}x^{\alpha+1}\cosh{(\sqrt{2r}x)} -
\frac{\alpha+1}{\sqrt{2r}}I(x),
\]
with $I(x):=\int_0^x y^{\alpha}\cosh{(\sqrt{2r}y)} \,dy$. Therefore, from
(\ref{derivataFBCD}) we may write
\begin{eqnarray*}
\label{derivataFBCD2} \frac{d}{dx}b^{\beta-1}(x) & = & \biggl(
\frac{x^2(\alpha+ \beta)}{2\beta g^2(x)} \biggr)\frac{\sinh
{(\sqrt {2r}x)}}{x} \bigl[-(\alpha+1) \sqrt{2r}I(x) +
\sinh{(\sqrt {2r}x)}x^{\alpha}\bigr]
\\
& =: & \biggl(\frac{x^2(\alpha+
\beta
)}{2\beta g^2(x)} \biggr)\frac{\sinh{(\sqrt{2r}x)}}{x} T(x).
\end{eqnarray*}
Since $T(0)=0$ and $T'(x)=\alpha x^{\alpha-1}[\sinh{(\sqrt{2r}x)} -
x\sqrt{2r}\cosh{(\sqrt{2r}x)}] =\break  - \alpha x^{\alpha+1}\psi'_r(x) < 0$,
being $x \mapsto\psi_r(x)$ increasing, it follows that $x \mapsto
T(x)$ is negative for any $x>0$. The decreasing property of $x \mapsto
b^{\beta-1}(x)$ is therefore proved.
\end{pf}

A computer drawing of the free-boundary
(\ref{FBBesselCD}) is provided in Figure \ref{figuraBessel}.

\begin{figure}

\includegraphics{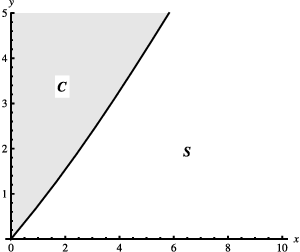}

\caption{A computer drawing of the free-boundary (\protect\ref{FBBesselCD})
when $r=\frac{1}{2}$ in the case of a three-dimensional Bessel process
and a Cobb--Douglas profit (with $\alpha=\beta=\frac{1}{2})$. The grey
area in the figure denotes the continuation (no-action) region, whereas
the white one denotes the stopping (action) region.}\vspace*{-5pt}
\label{figuraBessel}
\end{figure}

\subsubsection{CEV process}
\label{CEVD}

Let now the diffusion $X^{x}$ be of CEV (Constant Elasticity of
Variance) type, that is,
%
\begin{equation}
\label{CEV} dX^{x}(t) = r X^{x}(t) \,dt + \sigma
\bigl(X^{x}\bigr)^{1-\gamma}(t)\,dW(t),\qquad X^{x}(0) = x>0,
\end{equation}
for some $r > 0$, $\sigma>0$ and $\gamma\in(0,1/2]$. CEV process was
introduced in the financial literature by John Cox in $1975$ \cite{Cox}
in order to capture the stylized fact of a negative link between equity
volatility and equity price (the so-called ``leverage effect'').
In this case, we have
\[
m(dx)= \frac{2}{\sigma^2 x^{2(1-\gamma)}}e^{({r}/{(\gamma\sigma
^2)})x^{2\gamma}}\,dx,\qquad
 s(dx)=e^{-({r}/{(\gamma\sigma^2)})x^{2\gamma}}\,dx,
\]
and $\psi_r(x) =x$ (cf., e.g., \cite{DayaKar}, Section~6.2).

\begin{proposition}
For any $x >0$, one has
%
\begin{equation}
\label{FBCEV} b(x) = \biggl[\frac{2\beta}{\sigma^2(\alpha+ \beta)} g(x)
e^{-({r}/{(\gamma\sigma^2)})x^{2\gamma}}
\biggr]^{{1}/{(1-\beta)}},
\end{equation}
with $g(x):=\int_0^x y^{2\gamma+\alpha-1}e^{({r}/{(\gamma\sigma^2)})y^{2\gamma}}\,dy$.
\end{proposition}

\begin{pf}
From (\ref{integraleq}), one has
\[
\label{FBCEV1} \int_x^{\infty} \biggl(\int
_0^z y^{2\gamma+\alpha-1}e^{({r}/{(\gamma\sigma^2)})y^{2\gamma}}\,dy
\biggr) \frac{b^{\beta-1}(z)}{z^2}e^{-({r}/{(\gamma\sigma^2)})z^{2\gamma}}\,dz
 = \frac{\sigma^2}{x} \biggl(
\frac
{\alpha+ \beta}{2\beta} \biggr),
\]
that is,
\[
\label{FBCEV2} \int_x^{\infty} g(z)
\frac{b^{\beta-1}(z)}{z^2}e^{-({r}/{(\gamma\sigma^2)})z^{2\gamma}}\,dz = \frac{\sigma^2}{x} \biggl(
\frac{\alpha+
\beta
}{2\beta} \biggr),
\]
with $g(x):=\int_0^x y^{2\gamma+\alpha-1}e^{({r}/{(\gamma\sigma^2)})y^{2\gamma}}\,dy$.
Take now
\[
b^{\beta-1}(x)= \frac{\sigma^2}{g(x)} \biggl(\frac{\alpha+ \beta
}{2\beta
}
\biggr)e^{({r}/{(\gamma\sigma^2)})x^{2\gamma}}
\]
to obtain the desired result.

To complete the proof, we shall now show that $b(x)$ as in (\ref
{FBCEV}) is nondecreasing, or, equivalently, that $x \mapsto b^{\beta
-1}(x)$ is nonincreasing.
Indeed, we have
\begin{eqnarray*}
\label{FBCEVincr} \frac{d}{dx}b^{\beta-1}(x) & = & \frac
{\sigma^2}{g^2(x)}
\biggl(\frac{\alpha+ \beta}{2\beta} \biggr)x^{2\gamma
-1}e^{({r}/{(\gamma\sigma^2)})x^{2\gamma}} \biggl[
\frac{2r}{\sigma^2}g(x) - x^{\alpha}e^{({r}/{(\gamma\sigma^2)})x^{2\gamma}} \biggr]
\\
& = & -\frac{\alpha\sigma^2}{g^2(x)} \biggl(\frac{\alpha+ \beta}{2\beta} \biggr)x^{2\gamma-1}
e^{({r}/{(\gamma\sigma^2)})x^{2\gamma}}\int_0^x y^{\alpha-1 }
e^{({r}/{(\gamma\sigma^2)})y^{2\gamma}}
\,dy < 0,
\end{eqnarray*}
being $g(x)= \frac{\sigma^2}{2r}[e^{({r}/{(\gamma\sigma^2)})x^{2\gamma
}}x^{\alpha} - \alpha\int_0^x y^{\alpha-1 }e^{({r}/
{(\gamma\sigma^2)})y^{2\gamma}} \,dy]$, thanks to an integration by parts.
\end{pf}

A computer drawing of the free-boundary
(\ref{FBCEV}) is provided in Figure \ref{figuraCEV}.

\subsection{The case of a CES operating profit}
\label{CESsubsection}

In this section, we consider a nonseparable operating profit of CES
type, that is, $\pi(x,c)= (x^{{1}/{n}} + c^{{1}/{n}})^n$,
$(x,c) \in(0,\infty) \times(0,\infty)$ and $n\geq2$. Moreover, as in
the previous section, we take $X$ given by a geometric Brownian motion,
a three-dimensional Bessel process and a CEV process, respectively.

\begin{figure}

\includegraphics{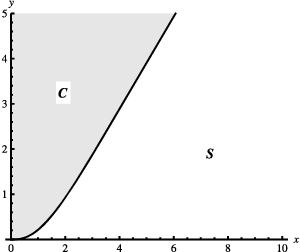}

\caption{A computer drawing of the free-boundary (\protect\ref
{FBCEV}) in the
case of a CEV process (with $\gamma=r=\frac{1}{2}$ and $\sigma=1$) and
a Cobb--Douglas profit (with $\alpha=\beta=\frac{1}{2}$). The grey area
in the figure denotes the continuation (no-action) region, whereas the
white one denotes the stopping (action) region.}\vspace*{-3pt}
\label{figuraCEV}
\end{figure}

Notice that CES operating profit does satisfy the first part of
Assumption~\ref{AssProfit} with $\kappa=1$ since $\lim_{c
\rightarrow
\infty}\pi_c(x,c)=\lim_{c \rightarrow\infty}[1 + (\frac
{x}{c})^{{1}/{n}}]^{n-1} = 1$. Then, according to Assumption~\ref
{rbiggerthankappa} here we take $r>1$.

Due to our identification $l^{*}(t)=b(X^{x}(t))$ [cf. (\ref
{identifico})], we expect that Assumption~\ref{rbiggerthankappa} might
play a role also in the optimal stopping problem (\ref{v}). Having
$r>1$ guarantees in fact that optimal stopping problem (\ref{v}) has a
nonempty continuation region.
Indeed, if the economic shock $X$ is a positive diffusion (as we will
consider in the examples below), one has
\begin{eqnarray*}
1 & \geq & \inf_{\tau\geq
0}\mathbb {E} \biggl\{\int
_0^{\tau}e^{-rs}\pi_c
\bigl(X^x(s),y\bigr)\,ds + e^{- r \tau
} \biggr\}
\nonumber
\\[-1pt]
& = & \inf_{\tau\geq0}\mathbb {E} \biggl\{\int_0^{\tau}e^{-rs}
\biggl[ \biggl(\frac{X^x(s)}{y} \biggr)^{1/n} + 1
\biggr]^{n-1}\,ds + e^{- r \tau} \biggr\}
\nonumber
\\[-1pt]
& \geq & \inf_{\tau\geq0}\mathbb {E} \biggl\{\int
_0^{\tau}e^{-rs} \biggl[(n-1) \biggl(
\frac{X^x(s)}{y} \biggr)^{1/n} + 1 \biggr]\,ds + e^{- r \tau}
\biggr\}
\nonumber
\\[-1pt]
& = & \frac{1}{r} + \inf_{\tau
\geq
0}\mathbb{E} \biggl\{\int
_0^{\tau}e^{-rs}(n-1) \biggl(
\frac
{X^x(s)}{y} \biggr)^{1/n} \,ds + \biggl(\frac{r-1}{r} \biggr)
e^{- r \tau} \biggr\},
\end{eqnarray*}
where we have used the generalized Bernoulli inequality for the third step.
If now $r \leq1$, the two terms in the last expected value above are
increasing functions of $\tau$ and, therefore, it is always optimal to
stop immediately, that is, $\tau^*=0$ for any $(x,y) \in(0,\infty)
\times(0,\infty)$, and thus $\mathcal{C}=\varnothing$.

Before starting with our examples, we also need a preliminary lemma
that will be useful in the following.
%

\begin{lemma}
\label{Dini}
Take $n \geq2$ and positive continuously differentiable functions $\{
\alpha_{k,n}\}_{1 \leq k \leq n-1}$ and $h$ on $(0,\infty)$. Then, for
any $x > 0$, the polynomial equation of order $n-1$ for the unknown $f_n(x)$
%
\begin{equation}
\label{equationDini1} \sum_{k=1}^{n-1} %
\pmatrix{ n-1 \vspace*{2pt}
\cr
k } %
\alpha_{k,n}(x)
f_n^k(x) - h(x) = 0,
\end{equation}
admits a unique positive solution. Moreover, $f_n(\cdot)$ is
continuously differentiable on $(0,\infty)$.
\end{lemma}

\begin{pf}
The existence of a unique positive solution to (\ref{equationDini1})
follows from a straightforward application of Descartes' rule of signs.
To show that such a solution $f_n(\cdot)$, $n \geq2$, is continuously
differentiable on $(0,\infty)$, define the function $\Phi_n\dvtx (0,\infty)
\times(0,\infty) \mapsto\mathbb{R}$ by
%
\begin{equation}
\label{Diniderivability} \Phi_n(x,y):=\sum_{k=1}^{n-1}
\pmatrix{ n-1 \vspace*{2pt}
\cr
k } %
\alpha_{k,n}(x)
y^k - h(x),\qquad n \geq2.
\end{equation}
By the first part of this proof, we already know that for any arbitrary
but fixed $x_o > 0$ there exists a unique positive $f_n(x_o)$ such that
$\Phi_n(x_o,f_n(x_o))=0$. Moreover, $\frac{\partial{\Phi
_n}}{\partial
y}(x_o,f_n(x_o))>0$ because $\alpha_{k,n}$ and $f_n$ are positive. Then
$f_n$ is continuously differentiable in a suitable neighborhood of
$x_o$, by the implicit function theorem. Since $x_o > 0$ was arbitrary,
it follows that $f_n$ is continuously differentiable on $(0,\infty)$.
\end{pf}

\subsubsection{Geometric Brownian motion}
As in Section~\ref{GBMCD}, let $X^{x}$ be a geometric Brownian motion
with drift $\mu\in\mathbb{R}$ and volatilty $\sigma>0$.

%
\begin{proposition}
\label{boundaryGBMCES}
Define $\delta:=\frac{\mu}{\sigma^2} - \frac{1}{2}$, $\gamma_1$
as the
positive root of the equation $\frac{1}{2}\sigma^2\gamma(\gamma-1) +
\mu\gamma= r$ and $\theta:=\gamma_1 + 2\delta$.
Then, for any $x>0$ and $n \geq2$ one has
%
\begin{equation}
\label{boundaryGBMCESformula} b_n(x)= \biggl(\frac{1}{C_n}
\biggr)^n x,
\end{equation}
where $C_n$ is the unique positive constant solving
%
\begin{equation}
\label{equation2F1} F_{2,1}\bigl(-(n-1), n\theta; n\theta+ 1;
-C_n\bigr) = r,
\end{equation}
with $F_{2,1}$ the ordinary hypergeometric function (see, e.g., \cite
{AbrSteg}, Chapter~15, for details).
\end{proposition}

\begin{pf}
From (\ref{integraleq}), one has
\begin{eqnarray*}
\label{CESGBM1} \frac{\sigma^2}{2}x^{-\gamma_1} & = & \int
_x^{\infty} \biggl[\int_0^z
\bigl(y^{{1}/{n}} + b_n^{{1}/{n}}(z)\bigr)^{n-1}y^{\theta- 1}\,dy
\biggr]b_n^{{1}/{n}-1}(z) z^{-\theta
-\gamma_1 - 1}\,dz
\\
& = &\int_x^{\infty} \biggl[\int
_0^{z/b_n(z)} \bigl( 1 + t^{{1}/{n}}
\bigr)^{n-1}t^{\theta- 1}\,dt \biggr] b_n^{\theta}(z)
z^{-\theta-\gamma_1 - 1}\,dz
\\
& = & \int_x^{\infty} \biggl[\int
_0^{g_n(z)} \bigl( 1 + t^{{1}/{n}}
\bigr)^{n-1}t^{\theta- 1}\,dt \biggr] g_n^{-\theta}(z)
z^{-\gamma_1 - 1}\,dz,
\end{eqnarray*}
where we have performed the change of variable $t:= y/b_n(z)$ and we
have defined $g_n(z):=z/b_n(z)$ and $\theta:=\gamma_1 + 2\delta$.
But
\[
\int_0^{g_n(z)} \bigl( 1 + t^{{1}/{n}}
\bigr)^{n-1}t^{\theta- 1}\,dt = \frac
{1}{\theta}g_n^{\theta}(z)F_{2,1}
\bigl(-(n-1), n\theta; n\theta+ 1, -g_n^{{1}/{n}}(z)\bigr),
\]
where $F_{2,1}$ is the ordinary hypergeometric function (cf. \cite
{AbrSteg}, Chapter $15$) and, therefore,
%
\begin{equation}
\label{CESGBM2} \int_x^{\infty} F_{2,1}
\bigl(-(n-1), n\theta; n\theta+ 1, -g_n^{{1}/{n}}(z)\bigr)
\gamma_1 z^{-\gamma_1 - 1}\,dz = \frac{\gamma_1\theta
\sigma
^2}{2}x^{-\gamma_1}.
\end{equation}
Take now $g_n^{{1}/{n}}(z) = C_n$ to be constant and notice that
$\frac{\gamma_1\theta\sigma^2}{2} = r$ to obtain that (\ref{CESGBM2})
is satisfied for $C_n$ solving
%
\begin{equation}
\label{CESGBM3} F_{2,1}\bigl(-(n-1), n\theta; n\theta+ 1,
-C_n\bigr) = r.
\end{equation}
According to \cite{AbrSteg}, Chapter $15$, equation (15.4.1) at page
$561$, it is easy to verify that~(\ref{CESGBM3}) is equivalent to the
polynomial equation of order $n-1$
%
\begin{equation}
\label{CESGBM4} \sum_{k=1}^{n-1}
\frac{(1-n)_k (n\theta)_k}{(1 + n\theta)_k}(-1)^k \frac
{C_n^k}{k!} - (r-1) = 0,
\end{equation}
where $(\cdot)_k$ denotes the Pochhammer symbol. Notice that all the
coefficients of the polynomial in (\ref{CESGBM4}) are positive except
for that of order zero which is instead negative since $r>1$ (cf. Assumption~\ref{rbiggerthankappa}); then (\ref{CESGBM4}) admits a
unique positive solution by Descartes' rule of signs and (\ref
{boundaryGBMCESformula}) is finally obtained recalling that $g_n(z) = z/b_n(z)$.
\end{pf}

\subsubsection{Three-dimensional Bessel process}
\label{BesselCES}

As in Section~\ref{Bessel2}, let now $X^{x}$ be a three-dimensional
Bessel process.

%
\begin{proposition}
\label{boundaryBESSEL3DCES}
For any $0 \leq k \leq n-1$ and $n \geq2$, define the functions
$\alpha
_{k,n}\dvtx (0,\infty) \mapsto(0,\infty)$ by
%
\begin{equation}
\label{coefficientsBessel3D} \alpha_{k,n}(x):=\int_0^x
y^{1 + {k}/{n}}\sinh(\sqrt{2r}y)\,dy.
\end{equation}
Then, for any $x>0$ and $n \geq2$ the free-boundary $b_n(\cdot)$ is
given by
%
\begin{equation}
\label{boundaryBESSEL3DCESformula} b_n(x)= \biggl(\frac{1}{f_n(x)}
\biggr)^n,
\end{equation}
where $f_n(x)$ is the unique positive solution of the polynomial
equation of order $n-1$
%
\begin{equation}
\label{equationDiniBessel} \sum_{k=1}^{n-1} %
\pmatrix{ n-1 \vspace*{2pt}
\cr
k } %
\alpha_{k,n}(x)
f_n^k(x) = (r-1)\alpha_{0,n}(x),\qquad  x>0.
\end{equation}
Moreover, the mapping $x \mapsto b_n(x)$ is nondecreasing.
\end{proposition}

\begin{pf}
The integral equation (\ref{integraleq}) takes the form
%
\begin{eqnarray}
\label{BesselCESintegraleq1}&& \frac{1}{2}\frac{x}{\sinh(\sqrt{2r}x)}\nonumber \\
&&\qquad =  \int
_x^{\infty} \biggl[\int_0^z
\biggl(1 + \biggl(\frac
{y}{b_n(z)} \biggr)^{{1}/{n}} \biggr)^{n-1}y
\sinh(\sqrt {2r}y)\,dy \biggr]\frac{dz}{\sinh^2(\sqrt{2r}z)}
\\
\nonumber &&\qquad =  \int_x^{\infty}\sum
_{k=0}^{n-1} %
\pmatrix{ n-1 \vspace*{2pt}
\cr
k } %
\biggl[\int_0^z \biggl(
\frac{y}{b_n(z)} \biggr)^{{k}/{n}}y \sinh (\sqrt {2r}y)\,dy \biggr]
\frac{dz}{\sinh^2(\sqrt{2r}z)}
\\
&&\qquad =  \int_x^{\infty}\sum
_{k=0}^{n-1} %
\pmatrix{ n-1 \vspace*{2pt}
\cr
k } %
\alpha_{k,n}(z) f_n^{k}(z)
\frac{dz}{\sinh^2(\sqrt{2r}z)},\nonumber
\end{eqnarray}
where we have used the binomial expansion and the definitions (\ref
{coefficientsBessel3D}) and (\ref{boundaryBESSEL3DCESformula}).
Since
$\frac{1}{2}\frac{x}{\sinh(\sqrt{2r}x)} = r \int_x^{\infty} \frac
{\alpha
_{0,n}(z) \,dz}{\sinh^2(\sqrt{2r}z)}$, one easily has from (\ref
{BesselCESintegraleq1})
\[
\label{BesselCESintegraleq2} \int_x^{\infty}\sum
_{k=1}^{n-1} %
\pmatrix{ n-1 \vspace*{2pt}
\cr
k } %
\alpha_{k,n}(z) f_n^{k}(z)
\frac{dz}{\sinh^2(\sqrt{2r}z)} = (r-1)\int_x^{\infty}
\frac{\alpha_{0,n}(z) \,dz}{\sinh^2(\sqrt{2r}z)},
\]
for any $x>0$, and thus by differentiating
%
\begin{equation}
\label{polynomialBESSEL} \sum_{k=1}^{n-1} %
\pmatrix{ n-1 \vspace*{2pt}
\cr
k } %
\alpha_{k,n}(x)
f_n^{k}(x) = (r-1)\alpha_{0,n}(x),
\end{equation}
for a.e. $x>0$.
It now remains to show that (\ref{polynomialBESSEL}) actually admits at
most one positive solution and that (\ref{polynomialBESSEL}) holds for
every $x>0$. Existence of a unique positive solution is guaranteed by
Lemma~\ref{Dini} with $h(x):= (r-1)\alpha_{0,n}(x)$, which is positive
due to Assumption~\ref{rbiggerthankappa}. Moreover, Lemma~\ref{Dini}
also ensures that $f_n(\cdot)$ is continuously differentiable on
$(0,\infty)$ and, therefore, (\ref{polynomialBESSEL}) actually holds
for every $x>0$.

As for the nondecreasing property of $x \mapsto b_n(x)$, $n \geq2$,
because of (\ref{boundaryBESSEL3DCESformula}) it suffices to prove that
$x \mapsto f_n(x)$, $n \geq2$, is nonincreasing. First of all, it is
not hard to see that $x \mapsto f_2(x)$ is nonincreasing by direct
calculations. To prove that any $f_n$, $n > 2$, is nonincreasing as
well, we can proceed as follows.
Thanks to Lemma~\ref{Dini} we can differentiate (\ref
{polynomialBESSEL}) to obtain
\begin{eqnarray*}
&&\frac{f'_n(x)}{f_n(x)}\sum_{k=1}^{n-1}
\pmatrix{ n-1 \vspace*{2pt}
\cr
k } %
k
\alpha_{k,n}(x) f_n^k(x) \\
&&\qquad= (r-1)
\alpha_{0,n}^{\prime}(x)
 - \sum_{k=1}^{n-1} %
\pmatrix{ n-1 \vspace*{2pt}
\cr
k } %
\alpha_{k,n}^{\prime}(x)
f_n^k(x),
\end{eqnarray*}
from which
%
\begin{equation}
\label{diseguaglianzaBessel1} \quad\frac{f'_n(x)}{f_n(x)} \sum_{k=1}^{n-1}
\pmatrix{ n-1 \vspace*{2pt}
\cr
k } %
k
\alpha_{k,n}(x) f_n^k(x) \leq(r-1)
\alpha_{0,n}^{\prime}(x) - \alpha _{1,2}^{\prime}(x)f_2(x),
\end{equation}
because the coefficients $\alpha_{k,n}$ are nondecreasing and $f_n$ is
positive.
Noticing that
$f_2(x) = (r-1)\alpha_{0,n}(x)/ \alpha_{1,2}(x)$ and plugging it into
(\ref{diseguaglianzaBessel1}) we find
\[
\frac{f'_n(x)}{f_n(x)} \sum_{k=1}^{n-1}
\pmatrix{ n-1 \vspace*{2pt}
\cr
k } %
k
\alpha_{k,n}(x) f_n^k(x) \leq(r-1) \biggl[
\alpha_{0,n}^{\prime}(x) - \frac
{\alpha_{0,n}(x) \alpha_{1,2}^{\prime}(x)}{ \alpha_{1,2}(x)} \biggr].
\]
Since now $\alpha_{1,2}(x) \leq\sqrt{x}\alpha_{0,n}(x)$ and $\sqrt {x}\alpha_{0,n}^{\prime}(x) - \alpha_{1,2}^{\prime}(x) = 0$, by definition, then
\[
\frac{f'_n(x)}{f_n(x)} \sum_{k=1}^{n-1}
\pmatrix{ n-1 \vspace*{2pt}
\cr
k } %
k
\alpha_{k,n}(x) f_n^k(x) \leq\frac{(r-1)}{\sqrt{x}}
\bigl[\sqrt {x}\alpha _{0,n}^{\prime}(x) - \alpha_{1,2}^{\prime}(x)
\bigr] = 0,
\]
and the claimed nonincreasing property of $f_n(\cdot)$, $n \geq2$, follows.
\end{pf}

Notice that finding the free-boundary of a quite intricate nonseparable
singular control problem has been reduced to determine the positive
root of a polynomial equation. Clearly, that can be done analytically
up to the second order (i.e., $n=3$). Then, for higher orders,
mathematical softwares can help in solving such a simple computational problem.

\subsubsection{CEV process}
\label{CEVCES}

As in Section~\ref{CEVD}, let $X^{x}$ be a CEV (Constant Elasticity of
Variance) process of parameter $\gamma\in(0,\frac{1}{2}]$ [see (\ref{CEV})].
Exploiting arguments completely similar to those used in the proof of
Proposition~\ref{boundaryBESSEL3DCES} we can show the following.

%
\begin{proposition}
\label{boundaryCEVCES}
For any $0 \leq k \leq n-1$ and $n \geq2$, define the functions
$\alpha
_{k,n}\dvtx (0,\infty) \mapsto(0,\infty)$ by
%
\begin{equation}
\label{coefficientsCEV} \alpha_{k,n}(x):=\int_0^x
y^{2\gamma+ {k}/{n} - 1}e^{({r}/{(\gamma\sigma^2)})y^{2\gamma}}\,dy.
\end{equation}
Then, for any $x>0$ and $n \geq2$ the free-boundary $b_n(\cdot)$ is
given by
%
\begin{equation}
\label{boundaryCEVCESformula} b_n(x)= \biggl(\frac{1}{f_n(x)}
\biggr)^n,
\end{equation}
where $f_n(x)$ is the unique positive solution of the polynomial
equation of order $n-1$
%
\begin{equation}
\label{equationDiniCEV} \sum_{k=1}^{n-1} %
\pmatrix{ n-1 \vspace*{2pt}
\cr
k } %
\alpha_{k,n}(x)
f_n^k(x) = \frac{\sigma^2}{2} + (r-1)
\alpha_{0,n}(x),\qquad x>0.
\end{equation}
Moreover, the mapping $x \mapsto b_n(x)$ is nondecreasing.
\end{proposition}

\section*{Acknowledgments} I thankfully acknowledge two anonymous
referees for their pertinent and useful comments. I also thank Peter
Bank, Maria B. Chiarolla, Tiziano De Angelis, Salvatore Federico,
Goran Peskir, Frank Riedel and Jan-Henrik Steg for the constructive
discussions.
This paper was completed when I was visiting the Hausdorff Research
Institute for Mathematics (HIM) at the University of Bonn in the
framework of the Trimester Program ``Stochastic Dynamics in Economics
and Finance.'' I thank HIM for the hospitality.

%



\printaddresses

\end{document}